\documentclass[11pt]{article}
\usepackage{latexsym}
\usepackage{amssymb}
\usepackage{amsmath}
\usepackage{amsthm}
\usepackage[margin=1in]{geometry}
\usepackage[ruled,vlined,linesnumbered]{algorithm2e}
\usepackage{forloop}
\usepackage{lmodern}
\usepackage{paralist}
\usepackage{hyperref}
\usepackage{tikz}
\usepackage{etoolbox}
\usepackage{calc}
\usetikzlibrary{shapes,shapes.multipart}

\newtheorem{theorem}{Theorem}[section]
\newtheorem{lemma}[theorem]{Lemma}

\newtheorem{corollary}[theorem]{Corollary}
\newtheorem{definition}{Definition}[section]
\newtheorem{proposition}[theorem]{Proposition}
\newtheorem{observation}[theorem]{Observation}
\newtheorem{claim}[theorem]{Claim}

\newtheorem{problem}{Problem}
\newtheorem{hypothesis}{Hypothesis}

\makeatletter
\newtheorem*{rep@theorem}{\rep@title}
\newcommand{\newreptheorem}[2]{%
\newenvironment{rep#1}[1]{%
 \def\rep@title{#2 \ref{##1}}%
 \begin{rep@theorem}}%
 {\end{rep@theorem}}}
\makeatother

\newreptheorem{theorem}{Theorem}

\newcommand{\defcal}[1]{\expandafter\newcommand\csname c#1\endcsname{{\mathcal{#1}}}}
\newcounter{ct}
\forLoop{1}{26}{ct}{
    \edef\letter{\Alph{ct}}
    \expandafter\defcal\letter
}

\newcommand{\NP}{{\cal NP}}
\newcommand{\nonnegR}{\mathbb{R}^+}
\newcommand{\ee}{{\varepsilon}}
\newcommand{\ie}{{\it i.e.}}
\newcommand{\eg}{{\it e.g.}}
\newcommand{\DSS}[2]{{\cD_{#1}^+(#2)}}
\newcommand{\DS}[1]{{\cD^+(#1)}}
\newcommand{\DSF}[1]{{\cD^+_{#1}}}
\newcommand{\DDS}[2]{{\cD_{#1}(#2)}}
\newcommand{\DD}[1]{{\cD(#1)}}
\newcommand{\DDF}[1]{{\cD_{#1}}}
\newcommand{\Poly}{{\mathtt{Poly}}}
\newcommand{\ignore}[1]{}
\newcommand{\eqdef}{\stackrel{\mbox{\tiny{def}}}{=}}

\makeatletter
\def\moverlay{\mathpalette\mov@rlay}
\def\mov@rlay#1#2{\leavevmode\vtop{%
   \baselineskip\z@skip \lineskiplimit-\maxdimen
   \ialign{\hfil$\m@th#1##$\hfil\cr#2\crcr}}}
\newcommand{\charfusion}[3][\mathord]{
    #1{\ifx#1\mathop\vphantom{#2}\fi
        \mathpalette\mov@rlay{#2\cr#3}
      }
    \ifx#1\mathop\expandafter\displaylimits\fi}
\makeatother

\newcommand{\bigcupdot}{\charfusion[\mathop]{\bigcup}{\cdot}}

\title{\textbf{Constrained Monotone Function Maximization and the Supermodular Degree}}

\author{
Moran Feldman\thanks{School of Computer and Communications, EPFL.
Email: \texttt{moran.feldman@epfl.ch}.}
\and
Rani Izsak\thanks{Weizmann Institute of Science.
Email: \texttt{ran.izsak@weizmann.ac.il}.}
}

\begin{document}

\maketitle

\begin{abstract}
The problem of maximizing a constrained monotone set function has many practical applications and
generalizes many combinatorial problems such as $k$-\textsc{Coverage}, \textsc{Max-SAT}, \textsc{Set Packing}, \textsc{Maximum Independent Set} and \textsc{Welfare Maximization}.
Unfortunately, it is generally not possible to maximize a monotone set function up to an acceptable approximation ratio, even subject to simple constraints. One highly studied approach to cope with this hardness is to restrict the set function, for example, by requiring it to be submodular. An outstanding disadvantage of imposing such a restriction on the set function is that no result is implied for set functions deviating from the restriction, even slightly. 
A more flexible approach, studied by Feige and Izsak [ITCS 2013],
is to design an approximation algorithm whose approximation ratio depends on the complexity of the instance, as measured by some \textsf{complexity measure}.
Specifically, they introduced a complexity measure called \textsf{supermodular degree}, measuring deviation from submodularity, and designed an algorithm for the welfare maximization problem with an approximation ratio that depends on this measure. 

In this work, we give the first (to the best of our knowledge) algorithm for maximizing an {\em arbitrary} monotone set function, subject to a {\boldmath $k$}\textsf{-extendible system}. This class of constraints captures, for example, the intersection of $k$-matroids (note that a single matroid constraint is sufficient to capture the welfare maximization problem).
Our approximation ratio deteriorates gracefully with the complexity of the set function and $k$. 
Our work can be seen as generalizing both the classic result of Fisher, Nemhauser and Wolsey [Mathematical Programming Study 1978], for maximizing a submodular set function subject to a $k$-extendible system, and the result of Feige and Izsak for the welfare maximization problem.
Moreover, when our algorithm is applied to each one of these simpler cases, it obtains the same approximation ratio as of the respective original work. That is, the generalization does not incur any penalty.
Finally, we also consider the less general problem of maximizing a monotone set function subject to a uniform matroid constraint, and give a somewhat better approximation ratio for it.
\end{abstract}

\thispagestyle{empty}

\newpage

\pagenumbering{arabic}
\setcounter{page}{1}

\newtoggle{normal}
\newtoggle{SetS}
\forLoop{1}{6}{ct}{
	\expandafter\newtoggle{row\arabic{ct}}
	\expandafter\newtoggle{col\arabic{ct}}
	\expandafter\newtoggle{height\arabic{ct}}
}

\newcounter{tempa}
\newcounter{tempb}

\newcommand{\DrawTwoDimensionsRectangle}[3]{
	\ifdefequal{#3}{e}{
		\toggletrue{normal}
	}
	{
		\forLoop{1}{6}{ct}{
			\ifnumequal{#1}{\arabic{ct}}{\toggletrue{col\arabic{ct}}}{\togglefalse{col\arabic{ct}}}
			\ifnumequal{#2}{\arabic{ct}}{\toggletrue{row\arabic{ct}}}{\togglefalse{row\arabic{ct}}}
		}
		
		\ifboolexpr{(togl{row4} and togl{col1}) or (togl{row5} and togl{col2}) or (togl{row6} and togl{col3}) or (togl{row1} and togl{col4})}
		{
			\togglefalse{normal}
			\toggletrue{SetS}
		}
		\ifboolexpr{togl{row1} and (togl{col1} or togl{col2} or togl{col3})}
		{
			\togglefalse{normal}
			\togglefalse{SetS}
		}
		{
			\toggletrue{normal}
		}
	}
	
	\iftoggle{normal}
	{
		\draw[fill=lightgray, very thick] (#1, #2) rectangle (1+#1, 1+#2)
	}
	{\iftoggle{SetS}
	{
		\draw[fill=black, very thick] (#1, #2) rectangle (1+#1, 1+#2)
	}
	{
		\draw[fill=white, very thick] (#1, #2) rectangle (1+#1, 1+#2)
	}}
}

\newcommand{\ExampleTwoDimensions}[1]{
\begin{tikzpicture}[scale=0.47]
    \draw [<->,very thick] (0,8) node (yaxis) {}
        |- (8,0) node (xaxis) {};
		\foreach \x in {0,1,2,3,4,5}
     		\draw node[anchor=north] at (\x + 1.5, 0) {\x};
		\foreach \y in {0,1,2,3,4,5}
     		\draw node[anchor=east] at (0, \y + 1.5) {\y};
		
		\foreach \x in {1, 2, 3, 4, 5, 6}
			\DrawTwoDimensionsRectangle{\x}{1}{#1};

		\foreach \x in {1, 2, 3}
			\foreach \y in {4, 5, 6}
				\DrawTwoDimensionsRectangle{\x}{\y}{#1};

\end{tikzpicture}
}

\newcommand{\DrawCube}[4]{
	\ifdefequal{#4}{e}{
		\toggletrue{normal}
	}
	{
		\forLoop{1}{6}{ct}{
			\ifnumequal{#1}{\arabic{ct}}{\toggletrue{col\arabic{ct}}}{\togglefalse{col\arabic{ct}}}
			\ifnumequal{#2}{\arabic{ct}}{\toggletrue{row\arabic{ct}}}{\togglefalse{row\arabic{ct}}}
			\ifnumequal{#3}{\arabic{ct}}{\toggletrue{height\arabic{ct}}}{\togglefalse{height\arabic{ct}}}
		}
		
		\ifboolexpr{(togl{height1} and togl{row2} and togl{col3}) or (togl{height6} and togl{row3} and togl{col2}) or (togl{height5} and togl{row4} and togl{col1}) or
								(togl{height4} and togl{row5} and togl{col6}) or (togl{height3} and togl{row6} and togl{col5})}
		{
			\togglefalse{normal}
			\toggletrue{SetS}
		}
		\ifboolexpr{togl{height1} and ((togl{col1} and togl{row6}) or (togl{col2} and togl{row5}))}
		{
			\togglefalse{normal}
			\togglefalse{SetS}
		}
		{
			\toggletrue{normal}
		}
	}
	
	\iftoggle{normal}
	{
		\draw[fill=lightgray, very thick] (0.3*#2+#1, -0.5+0.3*#2+#3) rectangle (1+0.3*#2+#1, 0.5+0.3*#2+#3)
																			node[fill=lightgray, very thick, draw, trapezium, trapezium left angle=45, trapezium right angle=-45, minimum height=0.3cm, scale=0.8] at (0.3*#2+#1+0.75, 0.8+0.3*#2+#3) {}
																			node[fill=lightgray, very thick, draw, trapezium, trapezium left angle=-45, trapezium right angle=45, minimum height=0.3cm, scale=0.78, rotate=90] at (0.3*#2+#1+1.3, 0.27+0.3*#2+#3) {}
	}
	{\iftoggle{SetS}
	{
		\draw[fill=black, very thick] (0.3*#2+#1, -0.5+0.3*#2+#3) rectangle (1+0.3*#2+#1, 0.5+0.3*#2+#3)
																			node[fill=black, very thick, draw, trapezium, trapezium left angle=45, trapezium right angle=-45, minimum height=0.3cm, scale=0.8] at (0.3*#2+#1+0.75, 0.8+0.3*#2+#3) {}
																			node[fill=black, very thick, draw, trapezium, trapezium left angle=-45, trapezium right angle=45, minimum height=0.3cm, scale=0.78, rotate=90] at (0.3*#2+#1+1.3, 0.27+0.3*#2+#3) {}
	}
	{
		\draw[fill=white, very thick] (0.3*#2+#1, -0.5+0.3*#2+#3) rectangle (1+0.3*#2+#1, 0.5+0.3*#2+#3)
																			node[fill=white, very thick, draw, trapezium, trapezium left angle=45, trapezium right angle=-45, minimum height=0.3cm, scale=0.8] at (0.3*#2+#1+0.75, 0.8+0.3*#2+#3) {}
																			node[fill=white, very thick, draw, trapezium, trapezium left angle=-45, trapezium right angle=45, minimum height=0.3cm, scale=0.78, rotate=90] at (0.3*#2+#1+1.3, 0.27+0.3*#2+#3) {}
	}}
}

\newcommand{\ExampleThreeDimensions}[2]{
\begin{tikzpicture}[scale=0.47]
    \draw [<->,very thick] (0,8) node (yaxis) {}
        |- (8,0) node (xaxis) {};
		\draw [->,very thick] (0, 0) -- (3,3);
		\foreach \x in {0,1,2,3,4,5}
     		\draw node[anchor=north] at (1.05 * \x + 0.7, 0) {\x};
 		\draw node[anchor=east] at (1.3, 1.3) {5};
		\foreach \z in {0,1,2,3,4,5}
     		\draw node[anchor=east] at (0, \z + 1.3) {\z};

		\foreach \y in {6, 5, 4, 3, 2, 1}
			\foreach \x in {1, 2, 3, 4, 5, 6}
				\DrawCube{\x}{\y}{1}{#1};
				
		\foreach \z in {3, 4, 5, 6}
			\foreach \y in {6, 5, 4, 3, 2, 1}
				\foreach \x in {1, #2}
					\DrawCube{\x}{\y}{\z}{#1};

		\setcounter{tempa}{7-#2}
		\setcounter{tempb}{1+#2}
		\foreach \z in {3, 4, 5, 6}
			\foreach \y in {6, \arabic{tempa}}
				\foreach \x in {\arabic{tempb}, 3, 4, 5, 6}
					\DrawCube{\x}{\y}{\z}{#1};

\end{tikzpicture}
}
\section{Introduction} \label{sc:introduction}
A set function $f$ is a function assigning a non-negative real value to every subset of a given ground set $\cN$. A set function is (non-decreasing) monotone if $f(A) \leq f(B)$ whenever $A \subseteq B \subseteq \cN$. Monotone set functions are often used to represent utility/cost functions in economics and algorithmic game theory. From a theoretical perspective, many combinatorial problems such as $k$-\textsc{Coverage}, \textsc{Max-SAT}, \textsc{Set Packing} and \textsc{Maximum Independent Set} can be represented as constrained maximization of monotone set functions.

Unfortunately, it is generally not possible to maximize a general monotone set function up to an acceptable approximation ratio, even subject to simple constraints. For example, consider the case of a partition matroid constraint, where the ground set is partitioned into subsets of size $m$, and we are allowed to pick only a single element from each subset. This problem generalizes the well-known \textsc{welfare maximization} problem\footnote{The welfare maximization problem consists of a set $\cB$ of $m$ bidders and a set $\cN$ of $n$ items. Each bidder $b \in \cB$ has a monotone utility function $u_b : 2^\cN \rightarrow \mathbb{R}^+$. The objective is to assign a disjoint set $\cN_b \subseteq \cN$ of items to each bidder in a way maximizing $\sum_{b \in \cN} u_b(\cN_b)$ (\ie, the ``social welfare'').}, and thus, cannot be generally approximated by a factor of $O(\log m / m)$, in time polynomial in $n$ and $m$ (see Blumrosen and Nisan \cite{BN09}).\footnote{This result applies to value oracles, which we use throughout this work.}

One highly studied approach to cope with this hardness is to restrict the set function. 
A common restriction is submodularity. 
A set function is submodular if the marginal contribution of an element to a set can only decrease as the set increases. 
More formally, for every two sets $A \subseteq B \subseteq \cN$ and element $u \in \cN \setminus B$, $f(B \cup \{u\}) - f(B) \leq f(A \cup \{u\}) - f(A)$. Submodular functions are motivated by many real world applications since they represent the principle of economy of scale, and are also induced by many natural combinatorial structures (\eg, the cut function of a graph is submodular). Fortunately, it has been shown that submodular functions can be maximized, up to a constant approximation ratio, subject to various constraints. 
For example, maximizing a monotone submodular function subject to the partition matroid constraint, considered above, has a $(1-1/e)$-approximation algorithm (see Calinescu, Chekuri, Pal and Vondr\'{a}k \cite{CCPV11}).

An outstanding disadvantage of imposing a restriction on the set function, such as submodularity, is that no result is implied for functions deviating from the restriction, even slightly. A more flexible approach, studied by \cite{FI13}, is to define a \textsf{complexity measure} for set functions, and then design an approximation algorithm whose guarantee depends on this measure. 
More specifically, \cite{FI13} introduced a complexity measure called \textsf{supermodular degree}.
A submodular set function has a supermodular degree of~0. The supermodular degree becomes larger as the function deviates from submodularity.
Feige and Izsak \cite{FI13} designed a $(1/(d+2))$-approximation algorithm for the welfare maximization problem,
where $d$ is the maximum supermodular degree of the bidders' utility functions. 

In a classic work, Fisher, Nemhauser and Wolsey~\cite{FNW78} introduced a $(1/(k+1))$-approx\-imation algorithm for maximizing a submodular set function subject to a {\boldmath $k$}\textsf{-extendible system} (in fact, they proved this approximation ratio even for a more general class of constraints called {\boldmath $k$}-\textsf{systems}).
In this work, we leverage their work, together with the supermodular degree, and give the first (to the best of our knowledge) algorithm for maximizing an {\em arbitrary} monotone set function subject to a $k$-extendible system. 
Note that $k$-extendible system generalizes, for example, the intersection of $k$-matroids (see Section~\ref{sc:preliminaries} for definitions), and thus, also the welfare maximization problem, which can be captured by a single matroid constraint.
As in the works of Fisher, Nemhauser and Wolsey~\cite{FNW78} and~\cite{FI13}, our algorithm is greedy.
Like in \cite{FI13}, the approximation ratio of our algorithm deteriorates gracefully with the complexity of the set function.
Interestingly, when our algorithm is applied to the simpler cases studied by \cite{FNW78} and \cite{FI13}, its approximation ratio is exactly the same as that proved by the respective work. That is, we have no penalty for generality, either for handling an arbitrary set function (as opposed to only submodular) or for handling an arbitrary $k$-extendible system (as opposed to only welfare maximization).
We also show an hardness result, depending on $k$ and the supermodular degree of the instance,
suggesting the approximation ratio of our algorithm is almost the best possible.
Finally, we consider the less general problem of maximizing a monotone set function subject to a \textsf{uniform matroid} constraint 
(see Section~\ref{sc:preliminaries}), and give a somewhat better approximation ratio for it.

\subsection{Related work} \label{sc:related_work}
Extensive work has been conducted in recent years in the area of maximizing monotone submodular set functions subject to various constraints. We mention here the most relevant results.
Historically, one of the very first problems examined was maximizing a monotone submodular set function subject to a matroid constraint. Several special cases of matroids and submodular functions were studied in \cite{CC84,HK78,HKJ80,J76,KH78}, using the greedy approach. Recently, the general problem, with an arbitrary matroid and an arbitrary submodular set function, was given a tight approximation of $(1-1/e)$ by Calinescu et al. \cite{CCPV11}. A matching lower bound is due to \cite{NW78,NWF78}.

The problem of maximizing a monotone submodular set function over the intersection of $k$ matroids was considered by Fisher et al. \cite{FNW78}, who gave a greedy algorithm with an
approximation ratio of $1/(k+1)$, and stated that their proof extends to the more general class of $k$-systems using the outline of Jenkyns \cite{J76} (the extended proof is explicitly given by Calinescu et al. \cite{CCPV11}).  For $k$-intersection systems and $k$-exchange systems, this result was improved by Lee et al. \cite{LSV10} and Feldman et al. \cite{FNSW11}, respectively, to $1/(k+\ee)$, for every constant $\ee > 0$. The improvement is based on a local search approach that exploits exchange properties of the underlying combinatorial structure. Ward \cite{W12} further improved the approximation ratio for $k$-exchange systems to $2 / (k + 3 + \ee)$ using a non-oblivious local search. However, for maximizing a monotone submodular set function over $k$-extendible independence systems (and the more general class of $k$-systems), the current best known approximation is still $1/(k+1)$ \cite{FNW78}.

Other related lines of work deal with maximization of non-monotone submodular set functions (constrained or unconstrained) (see \cite{BFNS12,FNS11,V13} for a few examples) and minimization of submodular set functions \cite{GKTW10,GLS81,IN09,IO09}.

The welfare maximization problem (or combinatorial auction) is unique in the sense that it was studied in the context of many classes of utility (set) functions, including classes generalizing submodular set functions such as sub-additive \cite{F06} and fractionally sub-additive valuations \cite{DS06}. For many of these classes a constant approximation algorithm is known \cite{BM97,DNS05,F06,FV10,GS99} assuming access to a demand oracle, which given a vector of prices returns a set of elements maximizing the welfare of a player given these prices. However, when only a value oracle is available to the algorithm (\ie, the only access the algorithm has to the utility functions is by evaluating them on a chosen set) one cannot get a better than a polynomial approximation ratio, even for fractionally sub-additive valuations~\cite{DS06}. We are not aware of any other maximization subject to a constraint problem that was studied with respect to a non-submodular objective before our work.

\section{Preliminaries} \label{sc:preliminaries}
In this work, we consider set functions $f:2^{\cN} \to \nonnegR$ that are (non-decreasing) monotone (\ie, $A \subseteq B \subseteq \cN$ implies $f(A) \leq f(B)$) and non-negative.
We denote the cardinality of $\cN$ by $n$.
For readability, given a set $S \subseteq \cN$ and an element $u \in \cN$ we use $S + u$ to denote $S \cup \{u\}$ and $S - u$ to denote $S \setminus \{u\}$. 

\subsection{Independence Systems}
Given a ground set $\cN$, a pair $(\cN, \cI)$ is called an \textsf{independence system} if $\cI \subseteq 2^\cN$ is hereditary (that is, for every set $S \in \cI$, every set $S' \subseteq S$ is also in $\cI$).
Independence systems are further divided into a few known classes. The probably most highly researched class of independence systems is the class of matroids.
\begin{definition} [Matroid]
An independence system is a \textsf{matroid} if for every two sets $S, T \in \cI$ such that $|S| > |T|$,
there exists an element $u \in S \setminus T$, such that $T + u \in \cI$. This property is called the \emph{augmentation property} of matroids.
\end{definition}

Two important types of matroids are uniform and partition matroids. In a \textsf{uniform matroid} a subset is independent if and only if its size is at most $k$, for some fixed $k$. In a \textsf{partition matroid}, the ground set $\cN$ is partitioned into multiple subsets $\cN_1, \cN_2, \dotsc, \cN_k$, and an independent set is allowed to contain at most a single element from each subset $\cN_i$.

Some classes of independence systems are parametrized by a value $k \in \mathbb{N}$ ($k \geq 1$). 
The following is a simple example of such a class.
\begin{definition} [$k$-intersection]
An independence system $(\cN, \cI)$ is a {\boldmath $k$}-\textsf{intersection} if
there exist $k$ matroids $(\cN, \cI_1) \ldots (\cN, \cI_k)$ such that a set $S \subseteq \cN$ is in $\cI$
if and only if $S \in \bigcap_{i = 1}^k \cI_i$. 
\end{definition}
The problem of $k$-dimensional matching can be represented as maximizing a linear function over a $k$-intersection independence system. In this problem, one looks for a maximum weight matching in a $k$-sided hypergraph, \ie, an hypergraph where the nodes can be partitioned into $k$ ``sides'' and each edge contains exactly one node of each side. The representation of this problem as the intersection of $k$ partition matroids consists of one matroid per ``side'' of the hypergraph. The ground set of such a matroid is the set of edges, and a subset of edges is independent if and only if no two edges in it share a common vertex of the side in question.

The following definition, introduced by Mestre \cite{M06}, describes a more general class of independence systems which is central to our work. 
\begin{definition} [$k$-extendible] \label{def:k_extendible}
An independence system $(\cN, \cI)$ is a {\boldmath $k$}\textsf{-extendible} system if for every two subsets $T \subseteq S \in \cI$ and element $u \not \in T$ for which $T \cup \{u\} \in \cI$, there exists a subset $Y \subseteq S \setminus T$ of cardinality at most $k$ for which $S \setminus Y + u \in \cI$. 
\end{definition}

The problem of maximizing a linear function over a $k$-extendible system captures the problem of $k$-set packing.\footnote{$k$-set packing is, in fact, already captured by a smaller class called $k$-exchange, defined by \cite{FNSW11}.}
In this problem, one is given a weighted collection of subsets of $\cN$, each of cardinality at most $k$, and seeks a maximum weight sub-collection of pairwise disjoint sets.
The corresponding $k$-extendible system is as follows.
The ground set contains the sets as elements.
The independent subsets are all subsets of pairwise disjoint sets.
Let us explain why this is a $k$-extendible system. Adding a set $S$ of size $k$ to an independent set $I$, 
while respecting disjointness, requires that every elements of $S$ is not contained in any other set of $I$.
On the other hand, since $I$ is independent, each element is contained in at most one set of $I$.
Therefore, in order to add $S$, while preserving disjointness, we need to remove up to $k$ sets from $I$, as required by Definition~\ref{def:k_extendible}.

The most general class of independence systems considered is given by Definition~\ref{def:k_system}.
The following definition is used to define it.
\begin{definition} [Base]
Given an independence system $(\cN, \cI)$ and a set $S \subseteq \cN$, we say that a set $B \subseteq S$ is a \textsf{base} of $S$ if $B \in \cI$ but $B + u \not \in \cI$ for every element $u \in S \setminus B$. Furthermore, if $S = \cN$, then we say that $B$ is a base of the set system itself, or simply, a base.
\end{definition}

\begin{definition} [$k$-system] \label{def:k_system}
An independence system $(\cN, \cI)$ is a {\boldmath $k$}\textsf{-system} if for every set $S \subseteq \cN$, the ratio between the sizes of the smallest and largest bases of $S$ is at most $k$.
\end{definition}

An example of a natural problem which can be represented by a $k$-system, but not by a $k$-extendible system is given by \cite{CCPV11}. The following (strict) inclusions can be shown to hold \cite{CCPV11}:
\[
	\text{matroids} \subset k\text{-intersection} \subset k\text{-extendible systems} \subset k\text{-systems}
	\enspace.
\]

\subsection{Degrees of dependency}
We use the following standard definition.
\begin{definition} [Marginal set function] \label{def:marginal}
Let $f:2^{\cN} \to \nonnegR$ be a set function and let $u \in \cN$.
The \textsf{marginal set function} of $f$ with respect to $u$, denoted by $f(u \mid \cdot)$
is defined as $f(u \mid S) \eqdef f(S + u) - f(S)$.
When the underlying set function $f$ is clear from the context, we sometimes call $f(u \mid S)$ the \textsf{marginal contribution} of $u$ to the set $S$.
For subsets $S, T \subseteq \cN$,
we also use the notation $f(T \mid S) \eqdef f(S \cup T) - f(S)$.
\end{definition}

We recall the definitions of the complexity measures used in this work (defined by \cite{FI13}).

\begin{definition} [Dependency degree] \label{def:dependency}
The \textsf{dependency degree of an element} $u \in \cN$ by $f$ is defined as the cardinality of the set $\DDS{f}{u} = \{v \in \cN \mid \exists_{S \subseteq \cN} f(u \mid S + v) \neq f(u \mid S)\}$, containing all elements whose existence in a set might affect the marginal contribution of $u$.
$\DDS{f}{u}$ is called the \textsf{dependency set} of $u$ by $f$. 
The \textsf{dependency degree of a function} $f$, denoted by $\DDF{f}$, is simply the maximum dependency degree of any element $u \in \cN$. 
Formally, $\DDF{f} = \max_{u \in \cN} | \DDS{f}{u} |$.
When the underlying set function is clear from the context, we sometimes omit it from the notations.
\end{definition}

Note that $0 \leq \DDF{f} \leq n-1$ for any set function $f$. $\DDF{f} = 0$ when $f$ is {\em linear}, and becomes larger as $f$ deviates from linearity.

\begin{definition} [Supermodular (dependency) degree] \label{def:supermodular}
The \textsf{supermodular degree of an element} $u \in \cN$ by $f$ is defined as the cardinality of the set $\DSS{f}{u} = \{v \in \cN \mid \exists_{S \subseteq \cN} f(u \mid S + v) > f(u \mid S)\}$, containing all elements whose existence in a set might \emph{increase} the marginal contribution of $u$. $\DSS{f}{u}$ is called the \textsf{supermodular dependency set} of $u$ by $f$. 
The \textsf{supermodular degree of a function} $f$, denoted by $\DSF{f}$, is simply the maximum supermodular degree of any element $u \in \cN$. Formally, $\DSF{f} = \max_{u \in \cN} | \DSS{f}{u} |$.
Again, when the underlying set function is clear from the context, we sometimes omit it from the notations.
\end{definition}

Note that $0 \leq \DSF{f} \leq \DDF{f} \leq n-1$ for any set function $f$. $\DSF{f} = 0$ when $f$ is {\em submodular}, and becomes larger as $f$ deviates from submodularity.

\subsection{Representing the input}

Generally speaking, a set function might assign $2^n$ different values for the subsets of a ground set of size $n$.
Thus, one cannot assume that any set function has a succinct (\ie, polynomial in $n$) representation.
Therefore, it is a common practice to assume access to a set function via oracles.
That is, an algorithm handling a set function often gets an access to an oracle that answers queries about the function, instead of getting an explicit representation of the function.
Arguably, the most basic type of an oracle is the \textsf{value oracle}, which given any subset of the ground set, returns the value assigned to it by the set function. Formally:
\begin{definition} [Value oracle]
\textsf{Value oracle} of a set function $f:2^{\cN} \to \nonnegR$ is the following:
\\\textsf{Input:} A subset $S \subseteq \cN$.
\\\textsf{Output:} $f(S)$.
\end{definition}

Similarly, since in a given independence system the number of independence subsets might be, generally, exponential in the size of the ground set,
it is common to use the following type of oracle.
\begin{definition} [Independence oracle]
\textsf{Independence oracle} of an independence system $(\cN, \cI)$ is the following:
\\\textsf{Input:} A subset $S \subseteq \cN$.
\\\textsf{Output:} A Boolean value indicating whether $S \in \cI$.
\end{definition}
Our algorithms use the above standard oracles. 
Additionally, in order to manipulate a function with respect to the dependency/supermodular degree,
we need a way to know what are the (supermodular) dependencies of a given element in the ground set.
Oracles doing so were introduced by \cite{FI13}, and were used in their algorithms for the welfare maximization problem.
Formally:
\begin{definition} [Dependency and Supermodular oracles]
\textsf{Dependency oracle} (\textsf{Supermodular oracle}) of a set function $f:2^{\cN} \to \nonnegR$ is the following:
\\\textsf{Input:} An element $u \in \cN$.
\\\textsf{Output:} The set $\DD{u}$ ($\DS{u}$) of the (supermodular) dependencies of $u$ with respect to $f$.
\end{definition}

\subsection{Our results}
Our main result is an algorithm for maximizing any monotone set function subject to a $k$-extendible system,
with an approximation ratio that degrades gracefully as the supermodular degree increases.
Note that our algorithm achieves the best known approximation ratios also for the more specific problems of welfare maximization \cite{FI13} and maximizing a monotone submodular function subject to a $k$-extendible system \cite{FNW78}.
\begin{theorem} \label{th:supermodular_extendible}
There exists a $(1/(k(\DSF{f} + 1) + 1))$-approximation algorithm of $\Poly(|\cN|, 2^{\DSF{f}})$ time complexity for the problem of maximizing a non-negative monotone set function $f$ subject to a $k$-extendible system.
\end{theorem}

Note that an exponential dependence in $\DSF{f}$ is unavoidable, since, otherwise, we would get a polynomial time ($n-1$)-approximation algorithm for maximizing any set function subject to a $k$-extendible system.\footnote{To see that this cannot be done, consider the problem of maximizing the following family of set functions subject to a uniform $k=n/2$ matroid constraint. 
Each function in the family has a value of $1$ for sets strictly larger than $k$ and for a single set $A$ of size $k$. For all other sets the function assigns the value of $0$ (observe that $\DS{u} = \cN - u$ for every element $u \in \cN$, hence, the supermodular oracle is useless in this example). Given a random member of the above family, a deterministic algorithm using a polynomial number of oracle queries can determine $A$ only with an exponentially diminishing probability, and thus, will also output a set of value $1$ with such an exponentially diminishing probability. Using Yao's principle, this implies an hardness also for randomized algorithms.}

We show a similar result also for the dependency degree, providing a better approximation ratio when $\DDF{f} = \DSF{f}$.

\begin{theorem} \label{th:dependency_extendible}
There exists a $(1/(k(\DDF{f} + 1)))$-approximation algorithm of $\Poly(|\cN|, 2^{\DDF{f}})$ time complexity for the problem of maximizing a non-negative monotone set function $f$ subject to a $k$-extendible system.
\end{theorem}

On the other hand, we give tight examples for both algorithms guaranteed by Theorems~\ref{th:supermodular_extendible} and \ref{th:dependency_extendible}, and show the following hardness result via a reduction from $k$-dimensional matching.

\begin{theorem} \label{th:hardness}
No polynomial time algorithm for maximizing a non-negative monotone set function $f$ subject to a $k$-intersection independence system has an approximation ratio within $O\left(\frac{\log k + \log \DDF{f}}{k\DDF{f}}\right)$, unless $\cP = \NP$. This is true even if $k$ and $\DDF{f}$ are considered constants.
\end{theorem}

Note that since $\DSF{f} \leq \DDF{f}$ for any set function $f$, the hardness claimed in Theorem~\ref{th:hardness} holds also in terms of $\DSF{f}$.

Finally, we also consider the special case of a uniform matroid constraint, \ie, where one is allowed to pick an arbitrary subset of $\cN$ of size at most $k$. For this simpler constraint we present an algorithm whose approximation ratio has a somewhat better dependence on $\DSF{f}$.\footnote{The guarantee of Theorem~\ref{th:supermodular_uniform} is indeed an improvement over the guarantee of Theorem~\ref{th:dependency_extendible} for $1$-extendible system, because for every $x \geq 1$, $1 - e^{-1/x} \geq 1 - (1 - x^{-1} + x^{-2}/2) = x^{-1}(1 - x^{-1}/2) \geq x^{-1}/(1 + x^{-1}) = 1 / (x + 1)$.}

\begin{theorem} \label{th:supermodular_uniform}
There exists a $(1 - e^{-1 / (\DSF{f} + 1)})$-approximation algorithm of $\Poly(|\cN|, 2^{\DSF{f}})$ time complexity for the problem of maximizing a non-negative monotone set function $f$ subject to a uniform matroid constraint.
\end{theorem}

\begin{theorem} \label{th:hardness_uniform}
No polynomial time algorithm for maximizing a non-negative monotone set function $f$ subject to a uniform matroid constraint
has a constant approximation ratio, unless SSE (Small-Set Expansion Hypothesis)\footnote{See \cite{RS10, RST12} and Section~\ref{ap:uniform} for definitions and more information about SSE.} is false. 
\end{theorem}

\section{\texorpdfstring{$k$}{k}-Extendible system} \label{sc:extendible}
In this section we prove Theorems~\ref{th:supermodular_extendible} and~\ref{th:hardness}. The proof of Theorem~\ref{th:dependency_extendible} uses similar ideas and is deferred to Appendix~\ref{ap:dependency_extendile}, for readability.

\subsection{Algorithm for \texorpdfstring{$k$}{k}-extendible system (Proof of Theorem~\ref{th:supermodular_extendible})}
We consider in this section Algorithm~\ref{alg:ExtendibleSystemGreedy}, and prove it fulfills the guarantees of Theorem~\ref{th:supermodular_extendible}. 

\begin{algorithm}[h!t]
\caption{\textsf{Extendible System Greedy}$(f, \cI)$} \label{alg:ExtendibleSystemGreedy}
Initialize: $S_0 \leftarrow \varnothing$, $i \leftarrow 0$.\\
\While{$S_i$ is not a base}
{
    $i \leftarrow i + 1$.\\
    Let $u_i \in \cN \setminus S_{i - 1}$ and $D^+_{best}(u_i) \subseteq \DS{u_i}$ be a pair of an element and a set maximizing $f(D^+_{best}(u_i) + u_i ~|~ S_{i - 1})$ among all pairs obeying $S_{i - 1} \cup D^+_{best}(u_i) + u_i \in \cI.$ \label{ln:selection_extendible_system}\\
    $S_i \leftarrow S_{i - 1} \cup D^+_{best}(u_i) + u_i$.
}
Return $S_i$.
\end{algorithm}
First, let us give some intuition.
Let $APX$ be an approximate solution and let $OPT$ be an arbitrary optimal solution.
Originally, before the algorithm adds any elements to $APX$, 
it still can be that it chooses to add all elements of $OPT$ (together) to $APX$, and get an optimal solution.
At each iteration, when adding elements to $APX$, this possibility might get ruined for some elements of $OPT$. 
If we want to keep the invariant that all elements of $OPT$ can be added to $APX$, then we might 
have to discard some elements of $OPT$. 
This discard potentially decreases the value of $OPT$, and therefore, can be seen as the damage incurred by the iteration.
Note that by definition of a $k$-extendible system, we do not have to discard more than $k$ elements for every element we add.
That is, at every iteration, only up to $k(\DSF{F} + 1)$ elements must be discarded.
Therefore, if we manage to upper bound the damage of discarding a single element by the benefit of the allocation at the same iteration, we get the desired bound.\footnote{The other additive~1 in the denominator of the approximation ratio comes from the fact that by $OPT$'s value we actually mean its marginal contribution to $APX$. In this sense, addition of elements to $APX$ also might reduce the value of $OPT$.}
Recall that the supermodular dependencies of an element are exactly the elements that may increase its marginal value.
Therefore, when discarding an element from $OPT$, the maximum damage is bounded by the marginal value of this element with respect to its supermodular dependencies in $OPT$. But, as any subset of $OPT$ can be added to $APX$, the greedy choice of Algorithm~\ref{alg:ExtendibleSystemGreedy} explicitly takes into account the possibility of adding this element and its supermodular dependencies to $APX$. If another option is chosen, it must have at least the same immediate benefit, as wanted.

We now give a formal proof for Theorem~\ref{th:supermodular_extendible}.
Let us begin with the following observation.
\begin{observation} \label{ob:find_base}
Whenever $S_i$ is not a base, there exists an element $u \in \cN \setminus S_i$ for which $S_i \cup \varnothing + u \in \cI$ (note that $\varnothing \subseteq \DS{u}$). Hence, Algorithm~\ref{alg:ExtendibleSystemGreedy} always outputs a base.
\end{observation}

Throughout this section, we denote $d = \DSF{f}$.
Our proof is by a hybrid argument. That is, we have a sequence of hybrid solutions, one per iteration, where the first hybrid contains an optimal solution (and hence, has an optimal value), and the last hybrid is our approximate solution.\footnote{Actually, the last hybrid is defined as {\em containing} our approximate solution, but, as our approximate solution is a base, the hybrid must be exactly equal to it.}
Roughly speaking, we show the following: 
\begin{enumerate}
\item By adding each element to the approximate solution, we do not lose more than $k$ elements of the iteration`s hybrid 
(note that we add to our solution at most $d+1$ elements at any given iteration). 
This is formalized in Lemma~\ref{le:removed_items_count_extendible},
and the proof is based on Definition~\ref{def:k_extendible} ($k$-extendible system).
\item The damage from losing an element of an iteration's hybrid is bounded by the profit the algorithm gains at that iteration.
This is formalized in Lemma~\ref{le:tradeoff_iteration_extendible},
and the proof is based on Definition~\ref{def:supermodular} (supermodular degree).
\end{enumerate}
In conclusion, we show that when moving from one hybrid to the next, we lose no more than $k(d+1)$ times the profit at the respective iteration.

Let us formalize the above argument.
Let $\ell$ be the number of iterations performed by Algorithm~\ref{alg:ExtendibleSystemGreedy}, \ie, $\ell$ is the final value of $i$. 
We recursively define a series of $\ell + 1$ hybrid solutions as follows.
\begin{itemize}
    \item $H_0$ is a base containing $OPT$. By monotonicity, $f(H_0) = f(OPT)$.
    \item For every $1 \leq i \leq \ell$, $H_i$ is a maximum size independent subset of $H_{i - 1} \cup S_i$ containing $S_i$.
\end{itemize}

\begin{lemma} \label{le:removed_items_count_extendible}
For every iteration $1 \leq i \leq \ell$, $|H_{i - 1} \setminus H_i| \leq k \cdot |S_i \setminus H_{i - 1}| \leq k(d + 1)$.
\end{lemma}
\begin{proof}
Let us denote the elements of $S_i \setminus H_{i - 1}$ by $v_1, v_2, \ldots, v_r$. We prove by induction that there exists a collection of sets $Y_1, Y_2, \ldots, Y_r$, each of size at most $k$, such that: $Y_j \subseteq H_{i - 1} \setminus (S_{i - 1} \cup \{v_h\}_{h = 1}^{j - 1})$ and $H_{i - 1} \setminus (\cup_{h = 1}^j Y_h) \cup \{v_h\}_{h = 1}^j \in \cI$ for every $0 \leq j \leq r$. For ease of notation, let us denote $Y_1^j = \cup_{h = 1}^j Y_h$ and $v_1^j = \{v_h\}_{h = 1}^j$. Using this notation, the claim we want to prove can be rephrased as follows: there exists a collection of sets $Y_1, Y_2, \ldots, Y_r$, each of size at most $k$, such that: $Y_j \subseteq H_{i - 1} \setminus (S_{i - 1} \cup v_1^{j - 1})$ and $(H_{i - 1} \setminus Y_1^j) \cup v_1^j \in \cI$ for every $0 \leq j \leq r$.

For $j = 0$ the claim is trivial since $H_{i - 1} \in \cI$. Thus, let us prove the claim for $j$ assuming it holds for $j - 1$.
By the induction hypothesis, $(H_{i - 1} \setminus Y_1^{j-1}) \cup v_1^{j - 1} \in \cI$. On the other hand, $S_{i - 1} \cup v_1^{j - 1}$ is a subset of this set which is independent even if we add $v_j$ to it. Since $(\cN, \cI)$ is a $k$-extendible system, this implies the existence of a set $Y_j$ of size at most $k$ such that:
\[
	Y_j
	\subseteq
	[(H_{i - 1} \setminus Y_1^{j-1}) \cup v_1^{j - 1}] \setminus [S_{i - 1} \cup v_1^{j - 1}]
	\subseteq
	H_{i - 1} \setminus (S_{i - 1} \cup v_1^{j - 1})
	\enspace,
\]
and:
\[
	[(H_{i - 1} \setminus Y_1^{j-1}) \cup v_1^{j - 1}] \setminus Y_j + v_j \in \cI
	\Rightarrow
	(H_{i - 1} \setminus Y_1^j) \cup v_1^j \in \cI
	\enspace,
\]
which completes the induction step. Thus, $(H_{i - 1} \setminus Y_1^r) \cup v_1^r \in \cI$ is a subset of $H_{i - 1} \cup S_i$ which contains $S_i$ and has a size of at least: $|H_{i - 1}| - rk + r$. On the other hand, $H_i$ is a maximum size independent subset of $H_{i - 1} \cup S_i$, and thus:
$
	|H_i|
	\geq
	|H_{i - 1}| - rk + r
$. 
Finally, all elements of $H_i$ belong also to $H_{i - 1}$ except, maybe, the elements of $S_i \setminus S_{i - 1}$. Hence,
\[
	|H_{i - 1} \setminus H_i|
	\leq
	|H_{i - 1}| - |H_i| + |S_i \setminus S_{i - 1}|
	\leq
	|H_{i - 1}| - \left( |H_{i - 1}| - rk + r \right) + r
	=
	rk
	\enspace.
\]
Lemma~\ref{le:removed_items_count_extendible} now follows, since $r \leq d+1$.
\end{proof}

The following lemma upper bounds the loss of moving from one hybrid to the next one.

\begin{lemma} \label{le:tradeoff_iteration_extendible}
For every iteration $1 \leq i \leq \ell$, $f(H_{i - 1}) - f(H_i) \leq k(d + 1) \cdot f(D^+_{best}(u_i) + u_i ~|~ S_{i - 1})$,
where $u_i$ and $D^+_{best}(u_i)$ are the greedy choices made by Algorithm~\ref{alg:ExtendibleSystemGreedy} at iteration $i$.
\end{lemma}
\begin{proof}
Order the elements of $H_{i - 1} \setminus H_i$ in an arbitrary order $v_1, v_2, \ldots v_r$, and let $\bar{H}_j = H_{i - 1} \setminus \{v_h ~|~ 1 \leq h \leq j\}$. For every $1 \leq j \leq r$,
\begin{align} \label{eq:bestMarginal}
    f(\DS{v_j} \cap \bar{H}_j + v_j ~|~ S_{i - 1})
    ={} &
    f(v_j ~|~ (\DS{v_j} \cap \bar{H}_j) \cup S_{i - 1}) + f(\DS{v_j} \cap \bar{H}_j ~|~ S_{i - 1}) \notag \\
    \geq{} &
    f(v_j ~|~ (\DS{v_j} \cap \bar{H}_j) \cup S_{i - 1})
    \geq
    f(v_j ~|~ \bar{H}_j \cup S_{i - 1})
    \enspace,
\end{align}
where the first inequality follows by monotonicity and the second by Definition~\ref{def:supermodular} (supermodular degree).
Specifically, the latter is correct, since the supermodular dependencies of an element are the only ones that can increase its marginal contribution.
Therefore, adding elements of $\bar{H}_j \setminus \DS{v_j}$ to a set can only decrease the marginal contribution of $v_j$ with respect to this set. 
Since $\bar{H}_j \cup S_{i - 1} = \bar{H}_{j - 1} \cup S_{i - 1} - v_j$, we get:
\begin{align*}
	&
	\sum_{j = 1}^r f(\DS{v_j} \cap \bar{H}_j \setminus S_{i - 1} + v_j ~|~ S_{i - 1}) = 
    \sum_{j = 1}^r f(\DS{v_j} \cap \bar{H}_j + v_j ~|~ S_{i - 1}) \\ \geq{} &
    \sum_{j = 1}^r f(v_j ~|~ \bar{H}_j \cup S_{i - 1}) =
    f(\bar{H}_0 \cup S_{i - 1}) - f(\bar{H}_r \cup S_{i - 1}) \geq
    f(H_{i - 1}) - f(H_i)
    \enspace,
\end{align*}
where the two equalities follow by Definition~\ref{def:marginal} (marginal set function);
the first inequality follows by~\eqref{eq:bestMarginal} and the last inequality holds since $\bar{H}_0 = H_{i - 1} \supseteq S_{i - 1}$ and $\bar{H}_r \cup S_{i - 1} \subseteq H_i$. 
Lemma~\ref{le:tradeoff_iteration_extendible} now follows by recalling that $r \leq k(d + 1)$ (by Lemma~\ref{le:removed_items_count_extendible}), and noticing that the pair $(v_j, \DS{v_j} \cap \bar{H}_j \setminus S_{i-1})$ is a candidate pair that Algorithm~\ref{alg:ExtendibleSystemGreedy} can choose on Line~\ref{ln:selection_extendible_system} for every element $v_j \in H_{i - 1} \setminus H_i$.
\end{proof}

\begin{corollary}
Algorithm~\ref{alg:ExtendibleSystemGreedy} is a $1/(k(d + 1) + 1)$-approximation algorithm.
\end{corollary}

\begin{proof}
Adding up Lemma~\ref{le:tradeoff_iteration_extendible} over $1 \leq i \leq \ell$, we get:
\begin{align*}
    k(d + 1) \cdot [f(S_\ell) - f(S_0)]
    ={} &
    k(d + 1) \cdot \sum_{i = 1}^\ell f(D^+_{best}(u_i) + u_i ~|~ S_{i - 1})\\
    \geq{} &
    \sum_{i = 1}^\ell [f(H_{i - 1}) - f(H_i)]
    =
    f(H_0) - f(H_\ell)
    \enspace.
\end{align*}

Note that $H_\ell = S_\ell$ because $S_\ell$ is a base, and therefore, every independent set containing $S_\ell$ must be $S_\ell$ itself. Recall also that $f(H_0) = f(OPT)$ and $f(S_0) \geq 0$. Plugging these observations into the previous inequality gives:
\[
    k(d + 1) \cdot f(S_\ell) \geq f(OPT) - f(S_\ell)
    \Rightarrow
    f(S_\ell) \geq \frac{f(OPT)}{k(d + 1) + 1}
    \enspace.
    \qedhere
\]
\end{proof}

\subsubsection{A Tight Example for Algorithm~\ref{alg:ExtendibleSystemGreedy}} \label{ssc:tight_supermodular}
In this section we present an example showing that our analysis of Algorithm~\ref{alg:ExtendibleSystemGreedy} is tight even when the independence system $(\cN, \cI)$ belongs to $k$-intersection (recall that any independence system that is $k$-intersection is also $k$-extendible, but not vice versa).

\begin{proposition} \label{p:tight_extendible}
For every $k \geq 1$, $d \geq 0$ and $\ee > 0$, there exists a $k$-intersection independence system $(\cN, \cI)$ and a function $f : 2^\cN \rightarrow \mathbb{R}^+$ with $\DSF{f} = d$ for which Algorithm~\ref{alg:ExtendibleSystemGreedy} produces a $(1 + \ee) / (k(d + 1) + 1)$ approximation.
\end{proposition}

The rest of this section is devoted for constructing the independence system guaranteed by Proposition~\ref{p:tight_extendible}. Let $\cT$ be the collection of all sets $T \subseteq \{1, 2, \ldots, k + 1\} \times \{0, 1, \ldots, (d + 1)(k + 1) - 1\}$ obeying the following properties:
\begin{itemize}
	\item For every $1 \leq i \leq k + 1$, there exists exactly one $x$ such that $T$ contains the pair $(i, x)$.
	\item At least one pair $(i, x)$ in $T$ has $x \leq d$.
	\item Let $x_{k+1}$ be such that $(k + 1, x) \in T$. Then $x_{k+1} = 0$ or $x_{k+1} > d$.
\end{itemize}

Intuitively, the first requirement means that we can view a set $T \in \cT$ as a point in a $(k+1)$-dimensional space. The other two requirements make some points illegal.
For example, for $k=1$ the space is a $2(d+1) \times 2(d+1)$ grid, and the legal points are the ones that are either in row~0 or in one of the rows $d+1$ to $2(d+1)-1$ and one of the columns $0$ to $d$. Two examples of $\cT$ can be seen in Figure~\ref{fig:supermodular_construction_examples}.

\begin{figure}[!ht]
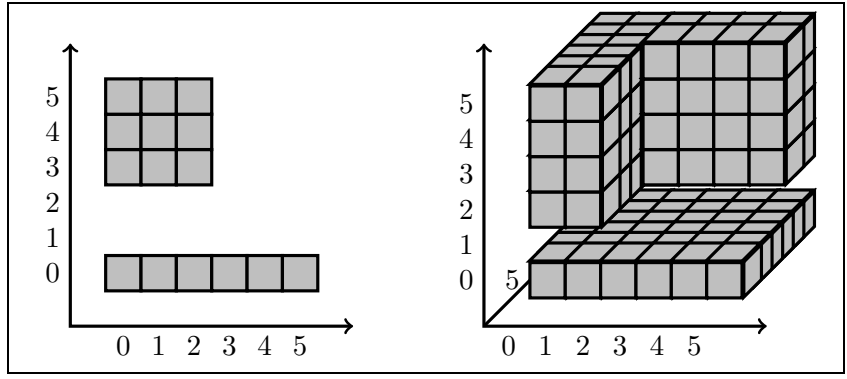

\center
\fbox{
	\ExampleTwoDimensions{e} \qquad \ExampleThreeDimensions{e}{2}
}
\caption{Graphical representations of $\cT$ for two configurations: $k=1,d=2$ and $k=2,d=1$. In both cases the last coordinate corresponds to the top-down axis.} \label{fig:supermodular_construction_examples}
\end{figure}

Let $\cN$ be the ground set $\{u_T \mid T \in \cT\}$. We define $k$ matroids on this ground set as follows. For every $1 \leq i \leq k$, $\cM_i = (\cN, \cI_i)$, where a set $S \subseteq \cN$ belongs to $\cI_i$ if and only if for every $0 \leq x < (d + 1)(k + 1)$, $|\{u_T \in S \mid (i, x) \in T\}| \leq 1$. One can easily verify that $\cM_i$ is a partition matroid. The independence system we construct is the intersection of these matroids, \ie, it is $(\cN, \cI)$, where $\cI = \bigcap_{i = 1}^k \cI_i$. Next, we define the objective function $f : 2^\cN \rightarrow \mathbb{R}^+$ as follows.
\[
	f'(S)
	=
	\sum_{x = 0}^{(d + 1)(k + 1) - 1} \mspace{-9mu} \min\{1, |\{u_T \in S \mid (k + 1, x) \in T\}|\}
	\enspace.
\]

That is, for $k=1$, $f'$ gains a value of~1 for every row that was ``hit'' by an element.
For every $0 \leq x \leq d$, let $\hat{T}(x) = \{(k + 1, 0)\} \cup \{(i, x)\}_{i = 1}^k$ (note that $\hat{T}_x \in \cT$).
\[
	f(S)
	=
	\begin{cases}
	f'(S) + \ee & \text{if $\{u_{\hat{T}(x)}\}_{x = 0}^d \subseteq S$} \enspace, \\
	f'(S) & \text{otherwise} \enspace.
	\end{cases}
\]
One can check that $f'$ is a non-negative monotone submodular function, and thus, $\DSF{f} = d$.

Claim~\ref{cl:algorithm_output} argues that Algorithm~\ref{alg:ExtendibleSystemGreedy} outputs a poor solution for the above independence system and objective function. The discussion after the claim presents an independent set $S^*$ of large value. Examples for both the solution of the algorithm and the set $S^*$ can be found in Figure~\ref{fig:supermodular_construction_examples_solutions}.

\begin{figure}
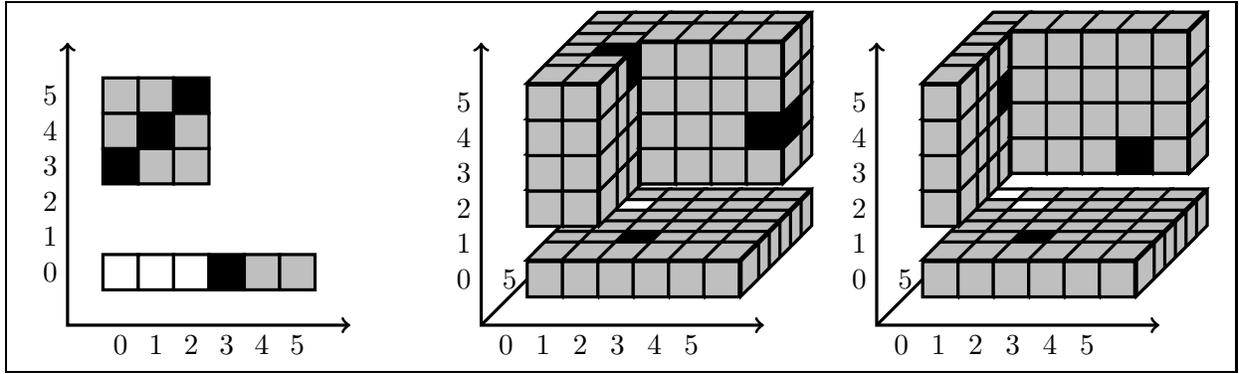

\center
\fbox{
	\ExampleTwoDimensions{s} \qquad \ExampleThreeDimensions{s}{2} \ExampleThreeDimensions{s}{1}
}
\caption{The solution produced by Algorithm~\ref{alg:ExtendibleSystemGreedy} and the set $S^*$ for the two examples presented in Figure~\ref{fig:supermodular_construction_examples}. The second example is depicted twice, once with some of the elements removed to make more elements visible. The set $S^*$ is denoted by black squares and the solution of Algorithm~\ref{alg:ExtendibleSystemGreedy} is denoted by white squares. Note that in both solutions no two elements share a row or a depth (when there is a depth). In $S^*$ no two points share a height, and thus, every element in $S^*$ contributes $1$ to the value. On the other hand, in the algorithm's solution all the elements share height, and thus its {\em overall value} is $1$ by $f'$ and $1+\ee$ by $f$.} \label{fig:supermodular_construction_examples_solutions}
\end{figure}

\begin{claim} \label{cl:algorithm_output}
Given the above constructed independence system $(\cN, \cI)$ and objective function $f$, Algorithm~\ref{alg:ExtendibleSystemGreedy} outputs a solution of value $1 + \ee$.
\end{claim}
\begin{proof}
Consider the first iteration of Algorithm~\ref{alg:ExtendibleSystemGreedy}.
Let $u_T \in \cN$. If $T \not \in \{\hat{T}(x)\}_{x=0}^d$, then $f(u_T \mid S) = f'(u_T \mid S)$ for every set $S \subseteq \cN$, and thus, $\DS{u_T} = \varnothing$ because $f'$ is a submodular function. Hence, for every such $u_T$, we get: $f(\DS{u_T} + u_T) = 1$. Consider now the case $T \in \{\hat{T}(x)\}_{x=0}^d$. In this case, clearly, $\DS{u_T} = \{u_{\hat{T}(x)}\}_{x=0}^d - u_T$, and thus, $f(\DS{u_T} + u_T) = 1 + \ee$. In conclusion, Algorithm~\ref{alg:ExtendibleSystemGreedy} picks exactly the elements of $\{u_{\hat{T}(x)}\}_{x=0}^d$ to its solution at the first iteration.

To complete the proof, we show that Algorithm~\ref{alg:ExtendibleSystemGreedy} cannot increase the value of its solution at the next iterations. Consider an arbitrary element $u_T \in \cN \setminus \{u_{\hat{T}(x)}\}_{x=0}^d$. By definition, $T$ must contain a pair $(i, x)$ such that $0 \leq x \leq d$. There are two cases:
\begin{itemize}
	\item If $i \neq k + 1$, then $u_T$ cannot coexist in an independent set of $\cM_i$ with $u_{\hat{T}(x)}$ because both correspond to sets containing the pair $(i, x)$.
	\item If $i = k + 1$, then $x = 0$ because $u_T \in \cN$.
\end{itemize}
From the above analysis, we get that all elements added to the solution after the first iteration contain the pair $(k + 1, 0)$ (and thus, no other pair of the form $(k + 1, x)$). Hence, they do not increase the value of either $f'$ or $f$.
\end{proof}

To prove Proposition~\ref{p:tight_extendible}, we still need to show that $(\cN, \cI)$ contains an independent set of a high value. Consider the set $S^* = \{u_{T^*(j)}\}_{j = 0}^{k(d + 1)}$, where $T^*(j) \eqdef \{(i, x) \mid 1 \leq i \leq k + 1 \text{ and } x = (i(d + 1) - j) \bmod (d + 1)(k + 1)\}$.

\begin{claim} \label{c:tightOptimal}
$S^* \subseteq \cN$.
\end{claim}
\begin{proof}
We need to show that for every $0 \leq j \leq k(d + 1)$, $u_{T^*(j)} \in \cN$. For $j = 0$, $(k + 1, 0) \in T^*(0)$, which completes the proof. Thus, we may assume from now on $1 \leq j \leq k(d + 1)$, and let $i = \lceil j / (d + 1) \rceil$.
Clearly $1 \leq i \leq k$ and $T^*(j)$ contains the pair $(i, x)$ for:
\[
	x
	=
	(i(d+1) - j) \bmod (d+1)(k + 1)
	=
	(\lceil j / (d+1) \rceil \cdot (d+1) - j) \bmod (d+1)(k + 1)
	\enspace.
\]
To conclude Claim~\ref{c:tightOptimal}, we need to show that $0 \leq x \leq d$. This follows since $\lceil j / (d+1) \rceil \cdot (d+1) - j \geq (j / (d+1)) \cdot (d+1) - j = 0$ and $\lceil j / (d+1) \rceil \cdot (d+1) - j < [j/(d+1) + 1] \cdot (d+1) - j = d + 1$.
\end{proof}

\begin{claim}
For every two values $0 \leq j_1 < j_2 \leq k(d + 1)$, $T^*(j_1) \cap T^*(j_2) = \varnothing$. 
Hence $S^* \in \cI$ and $f(S^*) \geq f'(S^*) = |S^*| = k(d + 1) + 1$.
\end{claim}
\begin{proof}
Assume towards contradiction that $(i, x) \in T^*(j_1) \cap T^*(j_2)$. Then, modulo $(d + 1)(k + 1)$, the following equivalence must hold:
\[
	(i(d + 1) - j_1) \equiv (i(d + 1) - j_2)
	\Rightarrow
	j_1 \equiv j_2
	\enspace,
\]
which is a contradiction since $j_1 \neq j_2$ and they are both in the range $[0, k(d + 1)]$.
\end{proof}

\subsection{Hardness (Proof of Theorem~\ref{th:hardness})}
Before proving Theorem~\ref{th:hardness} let us state the hardness result of \cite{HSS06} given by Theorem~\ref{th:hardness_matching}. In the $r$-\textsf{Dimensional Matching} problem one is given an $r$-sided hypergraph $G = (\bigcupdot_{i = 1}^r V_i, E)$, where every edge $e \in E$ contains exactly one vertex of each set $V_i$. The objective is to select a maximum size matching $M \subseteq E$, \ie, a subset $M \subseteq E$ of edges which are pairwise disjoint.

\begin{theorem}[Hazan et al. \cite{HSS06}] \label{th:hardness_matching}
It is \NP-hard to approximate $r$-\textsf{Dimensional Matching} to within $O(\log r/r)$ in polynomial time, even if $r$ is a constant.
\end{theorem}

Theorem~\ref{th:hardness} follows by combining Theorem~\ref{th:hardness_matching} with the following lemma.

\begin{lemma}
Any instance of $r$-\textsf{Dimensional Matching} can be represented as maximizing a monotone function $f$ with $\DSF{f} = \DDF{f} \leq d$ over a $k$-intersection set system for every $d \geq 0$ and $k \geq 1$ obeying $r \leq k(d + 1)$.
\end{lemma}
\begin{proof}
For simplicity, assume $r = k(d + 1)$. Let $G = (\bigcupdot_{i = 1}^r V_i, E)$ be the graph representing the $r$-\textsf{Dimensional Matching} instance. We first construct a new graph $G'$ as follows. For every edge $e \in E$ and $1 \leq j \leq d + 1$, let $e(j) = e \cap (\bigcupdot_{i = (j - 1)k + 1}^{jk} V_i)$, \ie, $e(j)$ is the part of $e$ hitting the vertex sets $V_{(j - 1)k + 1}, \dotsc, V_{jk}$. The edges of the new graph $G' = (\bigcupdot_{i = 1}^r V_i, E')$ are then defined as all edges that can be obtained this way. More formally:
\[
	E' = \{e(j) \mid e \in E \text{ and } 1 \leq j \leq d + 1\}
	\enspace.
\]

It is easy to see that the original instance of $r$-\textsf{Dimensional Matching} is equivalent to the problem of finding a matching in $G'$ maximizing the objective function $f : 2^{E'} \rightarrow \mathbb{R}^+$ defined as follows.
\[
	f(S) = \sum_{e \in E} \left \lfloor \frac{|S \cap \{e(j) \mid 1 \leq j \leq d + 1\}|}{d + 1} \right \rfloor
	\enspace.
\]
Moreover, $\DDF{f} = \DSF{f} = d$. Thus, to complete the proof we only need to show that the set of all legal matchings of $G'$ can be represented as a $k$-intersection independence system.

Consider the following partition of the vertices of $G'$. For every $1 \leq j \leq k$, $V'_j = \bigcupdot_{i = 0}^{d} V_{j + ki}$. Observe that each edge of $G'$ contains exactly one vertex of $V'_j$. Hence, the constraint that no two edges intersect on a node of $V'_j$ can be represented by the partition matroid $M_j = (E', \cI_j)$ defined as following. A set $S \subseteq E'$ is in $\cI_j$ if and only if no two edges of $S$ intersect on a node of $V'_j$. The set of legal matchings of $G'$ is, then, exactly $\bigcap_{j = 1}^k \cI_j$.
\end{proof}
\section{Uniform matroid constraint} \label{ap:uniform}
In this section we prove Theorems~\ref{th:supermodular_uniform} and~\ref{th:hardness_uniform}.

\subsection{Algorithm for uniform matroid constraint (Proof of Theorem~\ref{th:supermodular_uniform})}
Algorithm~\ref{alg:ExtendibleSystemGreedy} given in Section~\ref{sc:extendible} provides a $1 / (\DSF{f} + 2)$ approximation for a general $k$-extendible constraint. In this section, our objective is to improve over this approximation ratio for uniform matroid constraints. Throughout the section we use $d$ to denote $\DSF{f}$.

\subsubsection{First attempt}
Algorithm~\ref{alg:SimpleGreedy} is a slight simplification of Algorithm~\ref{alg:ExtendibleSystemGreedy} adapted to the context of a uniform matroid. We show that this algorithm already has a better than $1 / (\DSF{f} + 2)$ approximation ratio for some values of the parameters.

\begin{algorithm}[h!t]
\caption{\textsf{Simple Greedy}$(f, k)$} \label{alg:SimpleGreedy}
Initialize: $S_0 \leftarrow \varnothing$, $\ell = \lfloor k / (d + 1) \rfloor$.\\
\For{$i$ = $1$ \KwTo $\ell$}
{
    Let $u_i \in \cN$ be the element maximizing $f(\DS{u_i} + u_i ~|~ S_{i - 1})$.\\
    $S_i \leftarrow S_{i - 1} \cup \DS{u_i} + u_i$.
}
Return $S_\ell$.
\end{algorithm}

The feasibility of $S_\ell$ follows immediately by the observation that Algorithm~\ref{alg:SimpleGreedy} selects elements in $\ell \leq k / (d + 1)$ iterations, and in each iteration it selects up to $d + 1$ elements.

\begin{lemma} \label{le:simple_approximation}
Let $OPT$ be an arbitrary optimal solution.
Then, for every $0 \leq i \leq \ell$, $f(S_i) \geq [1 - (1 - 1/k)^i] \cdot f(OPT)$.
\end{lemma}
\begin{proof}
We prove the theorem by induction on $i$. For $i = 0$ the claim is trivial since $f(S_0) \geq 0 = [1 - (1 - 1/k)^0] \cdot f(OPT)$. Next, assume the claim holds for $i - 1$, and let us prove it for $i > 0$. Order the elements of $OPT$ in an arbitrary order $v_1, v_2, \ldots, v_k$ (by monotonicity, we may assume $|OPT| = k$), and let $OPT_j = \{v_h ~|~ 1 \leq h \leq j\}$. Then, for every element $v_j \in OPT$:
\begin{align*}
    f(\DS{v_j} + v_j ~|~ S_{i - 1})
    \geq{} &
    f(\DS{v_j} \cap OPT_{j - 1} + v_j ~|~ S_{i - 1})\\
    ={} &
    f(v_j ~|~ (\DS{v_j} \cap OPT_{j - 1}) \cup S_{i - 1}) + f(\DS{v_j} \cap OPT_{j - 1} ~|~ S_{i - 1})\\
    \geq{} &
    f(v_j ~|~ (\DS{v_j} \cap OPT_{j - 1}) \cup S_{i - 1})
    \geq
    f(v_j ~|~ OPT_{j - 1} \cup S_{i - 1})
    \enspace,
\end{align*}
where the first and second inequalities follow by monotonicity, and the last by Definition~\ref{def:supermodular} (supermodular degree). Summing up the above inequality over all elements of $OPT$, we get:
\[
    \sum_{j = 1}^k f(\DS{v_j} + v_j ~|~ S_{i - 1})
    \geq
    \sum_{j = 1}^k f(v_j ~|~ OPT_{j - 1} \cup S_{i - 1})
    =
    f(OPT \cup S_{i - 1}) - f(S_{i - 1})
    \geq
    f(OPT) - f(S_{i - 1})
    \enspace,
\]
where the first equality follows by Definition~\ref{def:marginal} (Marginal set function) and
the second inequality follows by monotonicity. Hence, there must exists an element $v \in OPT$ such that $f(\DS{v} + v ~|~ S_{i - 1}) \geq [f(OPT) - f(S_{i - 1})]/k$. Since $v$ is a potential candidate to be $u_i$ (\ie., the element selected by greedy choice of Algorithm~\ref{alg:SimpleGreedy}),
\begin{align*}
    f(S_i)
    ={} &
    f(\DS{u_i} + u_i ~|~ S_{i - 1}) + f(S_{i - 1})
    \geq
    \frac{f(OPT) - f(S_{i - 1})}{k} + f(S_{i - 1})\\
    ={} &
    \frac{f(OPT)}{k} + \frac{k - 1}{k} \cdot f(S_{i - 1})
    \geq
    \frac{f(OPT)}{k} + \frac{k - 1}{k} \cdot [1 - (1 - 1/k)^{i - 1}] \cdot f(OPT)\\
    = &
    [1 - (1 - 1/k)^i] \cdot f(OPT)
    \enspace,
\end{align*}
where the last inequality follows by induction hypothesis.
\end{proof}

\begin{corollary}
If $d + 1$ divides $k$, then the approximation ratio of Algorithm~\ref{alg:SimpleGreedy} is $1 - e^{-1/(d + 1)}$. Otherwise, the approximation ratio of Algorithm~\ref{alg:SimpleGreedy} is at least $1 - e^{-1/(d + 1)} - O(1/k)$.
\end{corollary}
\begin{proof}
Algorithm~\ref{alg:SimpleGreedy} outputs a set $S_\ell$, which by Lemma~\ref{le:simple_approximation} has a value of at least $[1 - (1 - 1/k)^\ell] \cdot f(OPT)$. If $d + 1$ divides $k$, then $\ell = k / (d + 1)$, and thus:
\[
	1 - (1 - 1/k)^\ell
	=
	1 - (1 - 1/k)^{k / (d + 1)}
	\geq
	1 - e^{-1/(d + 1)}
	\enspace.
\]
Otherwise, $\ell \geq k / (d + 1) - 1$, and thus:
\[
    1 - (1 - 1/k)^\ell
    \geq
    1 - (1 - 1/k)^{k/(d + 1) - 1}
    \geq
    1 - e^{-1/(d + 1)}e^{1/k}
		\geq
		1 - e^{-1/(d + 1)} - 2/k
    \enspace.
    \qedhere
\]
\end{proof}

Algorithm~\ref{alg:SimpleGreedy}, obviously, behaves very poorly when $k < d + 1$. However, this can be easily fixed by adding an additional phase to the algorithm as described by Algorithm~\ref{alg:SimpleGreedyFixed}.

\begin{algorithm}[h!t]
\caption{\textsf{Simple Greedy}$(f, k)$} \label{alg:SimpleGreedyFixed}
Initialize: $S_0 \leftarrow \varnothing$, $\ell = \lfloor k / (d + 1) \rfloor$.\\
\For{$i$ = $1$ \KwTo $\ell$}
{
    Let $u_i \in \cN$ be the element maximizing $f(\DS{u_i} + u_i ~|~ S_{i - 1})$.\\
    $S_i \leftarrow S_{i - 1} \cup \DS{u_i} + u_i$.
}
\If{$|S_{\ell}| = k$}
{
	Return $S_\ell$.\\
}
\Else
{
	Let $u_{\ell + 1}$ and $D^+_{best}(u_{\ell+1}) \subseteq \DS{u_{\ell + 1}}$ be a pair of an element and a set of size at most $k - |S_\ell|$ maximizing $f(D^+_{best}(u_{\ell+1}) + u_{\ell + 1} \mid S_{\ell})$.\\
	Return $S_\ell \cup D^+_{best}(u_{\ell+1}) + u_{\ell + 1}$.
}
\end{algorithm}

We are not aware of any example showing that the approximation ratio of Algorithm~\ref{alg:SimpleGreedyFixed} is worse than $1 - e^{-1/(d + 1)}$. We leave the problem of either finding such an example or improving the analysis of Algorithm~\ref{alg:SimpleGreedyFixed} as an open problem.

\subsubsection{Better approximation ratio}
In this section we analyse a more involved variant of Algorithm~\ref{alg:SimpleGreedy} achieving the approximation ratio of $1 - e^{-1/(d + 1)}$ guaranteed by Theorem~\ref{th:supermodular_uniform}.

Fix an arbitrary optimal solution $OPT$ of size $k$ (such an optimal solution exists by monotonicity), and let $d'$ be the maximum size of $\DS{u} \cap OPT$ for every $u \in OPT$. One can check the proof of Algorithm~\ref{alg:SimpleGreedy} and verify that if every reference to $d$ in the algorithm is replaced by $d'$ and $d' + 1$ happens to divide $k$, then the resulting algorithm has an approximation ratio of $1 - e^{-1/(d' + 1)} \geq 1 - e^{-1/(d + 1)}$.

An algorithm can guess\footnote{By ``guess'' we mean exhaustive search.} $d'$. 
To make sure that $d' + 1$ divides $k$ it might be necessary to modify $k$ by guessing some of the elements of $OPT$. Algorithm~\ref{alg:GuessGreedy} implements these ideas.
\begin{algorithm}[h!t]
\caption{\textsf{Guess Greedy}$(f, k)$} \label{alg:GuessGreedy}
Guess: $d'$, an element $u^*$ for which $d' = |\DS{u^*} \cap OPT|$ and the set $C = \DS{u^*} \cap OPT$ itself.\\
Initialize: let $r = k \bmod (d' + 1)$, $\ell = (k - r) / (d' + 1)$  and $S_0 \subseteq C$ be an arbitrary subset of size $r$.\\
\For{$i$ = $1$ \KwTo $\ell$}
{
    Let $u_i \in \cN$ and $D^+_{best}(u_i) \subseteq \DS{u_i}$ be a pair of an element and a set of size at most $d'$ maximizing $f(D^+_{best}(u_i) + u_i ~|~ S_{i - 1})$. \label{ln:selection}\\
    $S_i \leftarrow S_{i - 1} \cup D^+_{best}(u_i) + u_i$.
}
Return $S_\ell$.
\end{algorithm}

\begin{observation}
The time complexity of Algorithm~\ref{alg:GuessGreedy} is polynomial in $n$ and $2^d$.
\end{observation}
\begin{proof}
It is easy to check that the algorithm uses only a polynomial time (in $n$ and $2^d$) for every given guess. Thus, we only need to bound the number of possible guesses. The algorithm has $n(d + 1)$ possible guesses for $d'$ and $u^*$. For every such guess, there are at most $2^d$ possible guesses for $C$.
\end{proof}

\begin{observation}
Algorithm~\ref{alg:GuessGreedy} returns a feasible solution.
\end{observation}
\begin{proof}
For every $1 \leq i \leq \ell$, $|S_i| - |S_{i - 1}| \leq d' + 1$ elements. Thus:
\[
	|S_\ell|
	\leq
	|S_0| + (d' + 1) \ell
	=
	r + (d' + 1) \cdot \frac{k - r}{d' + 1}
	=
	k
	\enspace.
	\qedhere
\]
\end{proof}

We turn our attention to analysing the approximation ratio of Algorithm~\ref{alg:GuessGreedy}. To simplify the notation, we define $k' \eqdef k - r$. Observe that $\ell = k' / (d' + 1)$.

\begin{lemma} \label{le:approximation_guess}
For every $0 \leq i \leq \ell$, $f(S_i) \geq (1 - 1/k')^i \cdot f(S_0) + [1 - (1 - 1/k')^i] \cdot f(OPT)$.
\end{lemma}
\begin{proof}
We prove the theorem by induction. For $i = 0$ the claim is trivial since $f(S_0) = (1 - 1/k')^0 \cdot f(S_0) + [1 - (1 - 1/k')^0] \cdot f(OPT)$. Next, assume the claim holds for $i - 1$, and let us prove it for $i > 0$. Observe that $S_0 \subseteq OPT$, and therefore, $|OPT \setminus S_0| = k'$. Order the elements of $OPT \setminus S_0$ in an arbitrary order $v_1, v_2, \ldots, v_{k'}$, and let $OPT_j = \{v_h ~|~ 1 \leq h \leq j\}$. Then, for every element $v_j \in OPT \setminus S_0$,
\begin{align*}
    f(\DS{v_j} \cap OPT + v_j ~|~ S_{i - 1})
    ={} &
    f(v_j ~|~ (\DS{v_j} \cap OPT_{j - 1}) \cup S_{i - 1}) + f(\DS{v_j} \cap OPT_{j - 1} ~|~ S_{i - 1})\\
    \geq{} &
    f(v_j ~|~ (\DS{v_j} \cap OPT_{j - 1}) \cup S_{i - 1})
    \geq
    f(v_j ~|~ OPT_{j - 1} \cup S_{i - 1})
    \enspace,
\end{align*}
where the equality follows by Definition~\ref{def:marginal} (Marginal set function);
the first inequality follows by monotonicity and the second by Definition~\ref{def:supermodular} (supermodular degree). 
Summing up the above inequality over all elements of $OPT \setminus S_0$, we get:
\begin{align*}
		\sum_{j = 1}^{k'} f(\DS{v_j} \cap OPT + v_j \mid S_{i - 1})
		\geq{} &
    \sum_{j = 1}^{k'} f(v_j ~|~ OPT_{j - 1} \cup S_{i - 1})\\
    ={} &
    f(OPT \cup S_{i - 1}) - f(S_{i - 1})
    \geq
    f(OPT) - f(S_{i - 1})
    \enspace,
\end{align*}
where the equality follows by Definition~\ref{def:marginal} (Marginal set function)
and the second inequality follows by monotonicity. 
Note that the pair $(v_j, \DS{v_j} \cap OPT)$ is a candidate pair to be selected as $(u_i, D^+_{best}(u_i))$ for every element $v_j \in OPT \setminus S_0$, since $d'=\max\limits_{u \in OPT} | \DS{u} \cap OPT |$.
Hence, $f(D^+_{best}(u_i) + u_i \mid S_{i - 1}) \geq [f(OPT) - f(S_{i - 1})]/k'$. Thus,
\begin{align*}
    f(S_i)
    ={} &
    f(D^+_{best}(u_i) + u_i ~|~ S_{i - 1}) + f(S_{i - 1})
    \geq
    \frac{f(OPT) - f(S_{i - 1})}{k'} + f(S_{i - 1})\\
    ={} &
    \frac{f(OPT)}{k'} + \frac{k' - 1}{k'} \cdot f(S_{i - 1})\\
    \geq{} &
    \frac{f(OPT)}{k'} + \frac{k' - 1}{k'} \cdot \left[(1 - 1/k')^{i - 1} \cdot f(S_0) + [1 - (1 - 1/k')^{i-1}] \cdot f(OPT)\right]\\
    ={} &
    (1 - 1/k')^i \cdot f(S_0) + [1 - (1 - 1/k')^i] \cdot f(OPT)
    \enspace,
\end{align*}
where the last inequality follows by induction hypothesis.
\end{proof}

\begin{corollary}
The approximation ratio of Algorithm~\ref{alg:GuessGreedy} is $1 - e^{-1/(d' + 1)} \geq 1 - e^{-1/(d + 1)}$.
\end{corollary}
\begin{proof}
By Lemma~\ref{le:approximation_guess}, Algorithm~\ref{alg:GuessGreedy} outputs a set $S_\ell$ obeying:
\begin{align*}
    f(S_\ell)
    \geq{} &
    (1 - 1/k')^\ell \cdot f(S_0) + [1 - (1 - 1/k')^\ell] \cdot f(OPT)\\
    \geq{} &
    [1 - (1 - 1/k')^{k'/(d' + 1)}] \cdot f(OPT)
    \geq
    [1 - e^{-1/(d' + 1)}] \cdot f(OPT)
    \enspace.
    \qedhere
\end{align*}
\end{proof}

\subsection{Hardness (Proof of Theorem~\ref{th:hardness_uniform})}
The \textsc{Gap Small-Set Expansion problem} (introduced by \cite{RS10}) is the following promise problem.
\begin{problem} [\textsc{Gap Small-Set Expansion}($\eta, \delta$)] \label{prob:sse}
\ \\\textsf{Input:} An undirected graph $G=(V,E)$.
\\\textsf{Output:} Is $\phi_G(\delta) \geq 1-\eta$ or $\phi_G(\delta) \leq \eta$?
($\phi_G(\delta)$ is the edge expansion of $G$ with respect to subsets of size exactly $\delta |V|$.)
\end{problem}

The Small-Set Expansion Hypothesis (SSE), introduced by Raghavendra and Steurer \cite{RS10} (see, also, \cite{RST12}) is the following.
\begin{hypothesis}
For every $\eta > 0$, there exists $\delta$ such that Problem~\ref{prob:sse} parametrized by $\eta$ and $\delta$ is $\NP$-hard.
\end{hypothesis}

A hypergraph representation $F = (V_F, E_F, w_f)$ of a set function $f: 2^{\cN} \to \nonnegR$ (defined by \cite{CEEM08, CSS05}) is the following. The set $V_F$ contains exactly a single vertex for each element of the ground set $\cN$. The set $E_F$ is the set of the hyperedges of the hypergraph. The function $w_F : E_F \rightarrow \mathbb{R}$ assigns a real value for each hyperedge $e \in E_F$, and these values obey the following property: for every set $S \subseteq \cN$, the sum of the values of the hyperedges in the hypergraph induced by the vertices representing $S$ is exactly $f(S)$.
It is well known that any set function (normalized to have $f(\varnothing)=0$) can be uniquely represented by a hypergraph representation and vice versa.

Using the above definition, we show that maximizing a monotone set function subject to a uniform matroid constraint captures Problem~\ref{prob:sse}. Let $G=(V,E)$ be an arbitrary instance of Problem~\ref{prob:sse} with parameters $\eta$ and $\delta$.
We construct from $G$ an hypergraph representation $F=(V_F, E_F, w_F)$. The sets $V_F$ and $E_F$ are chosen as identical to $V$ and $E$, respectively (\ie, all hyperedges are of rank~2, hence, the hypergraph is in fact a graph). The value function $w_F$ gives a value of~1 for every edge of $E_F$. We can now consider the problem of finding a set of size at most $\delta |V|$ maximizing the set function $f$ corresponding to the hypergraph representation $F$. Theorem~\ref{th:hardness_uniform} follows immediately by the observation that a constant approximation for the last problem implies a constant approximation for Problem~\ref{prob:sse}.
\section{Future Research}
We view this work as a proof of concept showing that one can obtain interesting results for the problem of maximizing an arbitrary monotone set function subject to non-trivial constraints.
We would like to point out two possible directions for future research.
The first direction is studying the approximation ratio that can be guaranteed for more general problems as a function of the supermodular degree. Two possible such generalizations are a general $k$-system constraint and a non-monotone objective. Note that non-monotone objectives are interesting even in the unconstrained case.

The second direction is determining the guarantees that can be achieved for other complexity measures (with respect to either monotone or non-monotone set functions).
Specifically, we would like to draw attention to two complexity measures introduced by \cite{FFIILS14}, namely MPH (for monotone set functions) and PLE (for not necessarily monotone set functions). Both measures are based on fractionally sub-additive functions, a strict super-class of submodular functions, and they generally give lower values to set functions in comparison to the supermodular degree.
Thus, it is intriguing to show positive results for either of these measures.

\paragraph{Acknowledgments.}
Work of Moran Feldman is supported in part by ERC Starting Grant 335288-OptApprox.
Work of Rani Izsak is supported in part by the Israel Science Foundation (grant No. 621/12) and by the I-CORE Program of the Planning and Budgeting Committee and the Israel Science Foundation (grant No. 4/11).
We are grateful to Uri Feige and Irit Dinur for valuable discussions.
We are also grateful to Chidambaram Annamalai for his comments on a previous version of this manuscript.
\bibliographystyle{plain}
\bibliography{supermodular}

\appendix

\section{Proof of Theorem~\ref{th:dependency_extendible}} \label{ap:dependency_extendile}
We consider in this section Algorithm~\ref{alg:ExtendibleSystemGreedyDependency}, and prove it fulfills all the guarantees of Theorem~\ref{th:dependency_extendible}. 

\begin{algorithm}[h!t]
\caption{\textsf{Extendible System Greedy - Dependency Degree}$(f, \cI)$} \label{alg:ExtendibleSystemGreedyDependency}
Initialize: $S_0 \leftarrow \varnothing$, $i \leftarrow 0$.\\
\While{$S_i$ is not a base}
{
    $i \leftarrow i + 1$.\\
    Let $u_i \in \cN \setminus S_{i - 1}$ and $D_{best}(u_i) \subseteq \DD{u_i}$ be a pair of an element and a set maximizing $f(u_i ~|~ D_{best}(u_i) \cup S_{i - 1})$ among all pairs obeying $S_{i - 1} \cup D_{best}(u_i) + u_i \in \cI.$ \label{ln:selection_extendible_system_dependency}\\
    $S_i \leftarrow S_{i - 1} \cup D_{best}(u_i) + u_i$.
}
Return $S_i$.
\end{algorithm}

For the analysis of Algorithm~\ref{alg:ExtendibleSystemGreedyDependency} we use the same notation introduced in Section~\ref{sc:extendible}, except that we set $d = \DDF{f}$. Observe that both Observation~\ref{ob:find_base} and Lemma~\ref{le:removed_items_count_extendible} (and their proofs) apply also to Algorithm~\ref{alg:ExtendibleSystemGreedyDependency}. The following lemma is a counterpart of Lemma~\ref{le:tradeoff_iteration_extendible}.

\begin{lemma} \label{le:tradeoff_iteration_extendible_dependency}
For every iteration $1 \leq i \leq \ell$, $f(H_{i - 1}) - f(H_i) \leq [( k(d + 1) - 1 ] \cdot f(u_i ~|~ D_{best}(u_i) \cup S_{i - 1})$,
where $u_i$ and $D^+_{best}(u_i)$ are the greedy choices made by Algorithm~\ref{alg:ExtendibleSystemGreedy} at iteration $i$.
\end{lemma}

\begin{proof}
Order the elements of $H_{i - 1} \setminus H_i$ in an arbitrary order $v_1, v_2, \ldots v_r$, and let $\bar{H}_j = H_{i - 1} \setminus \{v_h ~|~ 1 \leq h \leq j\}$. 
By Definition~\ref{def:dependency} (dependency set), for every $1 \leq j \leq r$,
\[
    f(v_j ~|~ (\DD{v_j} \cap \bar{H}_j) \cup S_{i - 1})
    =
    f(v_j ~|~ \bar{H}_j \cup S_{i - 1})
    \enspace.
\]
Since $\bar{H}_j \cup S_{i - 1} = \bar{H}_{j - 1} \cup S_{i - 1} - v_j$, we get:
\begin{align} \label{eq:sum_bound}
    \sum_{j = 1}^r f(v_j ~|~ (\DD{v_j} \cap \bar{H}_j) \cup S_{i - 1})
    ={} &
    \sum_{j = 1}^r f(v_j ~|~ \bar{H}_j \cup S_{i - 1})\\ \nonumber
    ={} &
    f(\bar{H}_0 \cup S_{i - 1}) - f(\bar{H}_r \cup S_{i - 1})
    =
    f(H_{i - 1}) - f(\bar{H}_r)
    \enspace,
\end{align}
where the last equality holds since $\bar{H}_0 = H_{i - 1} \supseteq S_{i - 1}$ and $\bar{H}_r = H_{i - 1} \cap H_i \supseteq S_{i - 1}$.
We upper bound $f(\bar{H}_r)$ by recalling that $\bar{H}_r \subseteq H_i$, which gives by monotonicity $f(\bar{H}_r) \leq f(H_i)$ and then, by~\eqref{eq:sum_bound}, we have:
\[
	\sum_{j = 1}^r f(v_j ~|~ (\DD{v_j} \cap \bar{H}_j) \cup S_{i - 1})
	\geq
	f(H_{i - 1}) - f(H_i)
	\enspace.
\]
Note that the pair $(v_j, \left( \DD{v_j} \cap \bar{H}_j \right) \setminus S_{i - 1})$ is a candidate pair that Algorithm~\ref{alg:ExtendibleSystemGreedyDependency} can choose at Line~\ref{ln:selection_extendible_system_dependency} for every element $v_j \in H_{i - 1} \setminus H_i$. This implies the lemma, unless $r = k(d + 1)$ (recall that $r \leq k(d + 1)$ by Lemma~\ref{le:removed_items_count_extendible}).

Thus, we may assume from now on that $r = k(d + 1)$, which implies by Lemma~\ref{le:removed_items_count_extendible} that $|S_i \setminus H_{i - 1}| = d + 1$. In other words, the algorithm adds $u_i$ and all of $\DD{u_i}$ in the $i^{th}$ iteration. This means that the marginal contribution of $u_i$ is maximized when all of $\DD{u_i}$ is in the set, and thus, $u_i$ contributes to the hybrid solution the same value it contributes to the final solution. 
Formally, $|S_i \setminus H_{i - 1}|$ implies $D_{best}(u_i) = \DD{u_i}$ and $(\DD{u_i} + u_i) \cap H_{i - 1} = \varnothing$. Hence, since $\DD{u_i} + u_i \subseteq H_i$:
\begin{align*}
	f(H_i)
	= &
	f(\DD{u_i} + u_i \mid \bar{H}_r) + f(\bar{H}_r)\\
	\geq{} &
	f(u_i \mid \DD{u_i} \cup \bar{H}_r) + f(\bar{H}_r)
	=
	f(u_i \mid D_{best}(u_i) \cup S_{i - 1}) + f(\bar{H}_r)
	\enspace,
\end{align*}
where the inequality follows by monotonicity and the second equality by Definition~\ref{def:dependency} (dependency set) together with $D_{best}(u_i) = \DD{u_i}$.
Combining with \eqref{eq:sum_bound}, we get:
\[
	\sum_{j = 1}^{k(d + 1)} f(v_j ~|~ (\DD{v_j} \cap \bar{H}_j) \cup S_{i - 1}) - f(u_i \mid D_{best}(u_i) \cup S_{i - 1})
	\geq
	f(H_{i - 1}) - f(H_i)
	\enspace,
\]
which implies the lemma.
\end{proof}

\begin{corollary}
Algorithm~\ref{alg:ExtendibleSystemGreedyDependency} is a $(1/(k(d + 1)))$-approximation algorithm.
\end{corollary}
\begin{proof}
We have
\begin{align} \label{eq:dep:approximationRatio}
    [k(d + 1) - 1] \cdot [f(S_\ell) - f(S_0)]
    ={} &
    [k(d + 1) - 1] \cdot \sum_{i = 1}^\ell f(D_{best}(u_i) + u_i ~|~ S_{i - 1}) \\
		\geq{} &
		[k(d + 1) - 1] \cdot \sum_{i = 1}^\ell f(u_i ~|~ D_{best}(u_i) \cup S_{i - 1}) \nonumber \\
    \geq{} &
    \sum_{i = 1}^\ell [f(H_{i - 1}) - f(H_i)]
    =
    f(H_0) - f(H_\ell) \nonumber
    \enspace,
\end{align}
where the first inequality follows by monotonicity and the second by adding up
Lemma~\ref{le:tradeoff_iteration_extendible_dependency} over $1 \leq i \leq \ell$. 
Note that $H_\ell = S_\ell$ because $S_\ell$ is a base, and therefore, every independent set containing $S_\ell$ must be $S_\ell$ itself. Recall also that $f(H_0) = f(OPT)$ and $f(S_0) \geq 0$. Plugging these observations into~\eqref{eq:dep:approximationRatio} gives:
\[
    [k(d + 1) - 1] \cdot f(S_\ell) \geq f(OPT) - f(S_\ell)
    \Rightarrow
    f(S_\ell) \geq \frac{f(OPT)}{k(d + 1)}
    \enspace.
    \qedhere
\]
\end{proof}

\subsection{A tight example}
In this section we present an example showing that our analysis of Algorithm~\ref{alg:ExtendibleSystemGreedyDependency} is tight even when the independence system $(\cN, \cI)$ belongs to $k$-intersection (recall that any independence system in $k$-intersection is also $k$-extendible, but not vice versa).

\begin{proposition} \label{p:tight_extendible_dependency}
For every $k \geq 1$, $d \geq 0$ and $\ee > 0$, there exists a $k$-intersection independence system $(\cN, \cI)$ and a function $f : 2^\cN \rightarrow \mathbb{R}^+$ with $\DDF{f} = d$ for which Algorithm~\ref{alg:ExtendibleSystemGreedyDependency} produces a $(1 + \ee) / (k(d + 1))$ approximation.
\end{proposition}

The rest of this section is devoted for constructing the independence system guaranteed by Proposition~\ref{p:tight_extendible_dependency}. Let $\cT$ be the collection of all sets $T \subseteq \{1, 2, \ldots, k\} \times \{0, 1, \ldots, k(d + 1) - 1\}$ obeying the following properties:
\begin{itemize}
	\item For every $1 \leq i \leq k + 1$, there exists exactly one $x$ such that $T$ contains the pair $(i, x)$.
	\item At least one pair $(i, x)$ in $T$ has $x \leq d$.
\end{itemize}

Let $\cN$ be the ground set $\{u_T \mid T \in \cT\} \cup \{v_x\}_{x = 0}^{k(d + 1) - 1}$. We define $k$ matroids on this ground set as follows. For every $1 \leq i \leq k$, $\cM_i = (\cN, \cI_i)$, where a set $S \subseteq \cN$ belongs to $\cI_i$ if and only if for every $0 \leq x < k(d + 1)$, $|S \cap \{v_x\}| + |\{u_T \in S \mid (i, x) \in T\}| \leq 1$. One can easily verify that $\cM_i$ is a partition matroid. The independence system we construct is the intersection of these matroids, \ie, it is $(\cN, \cI)$, where $\cI = \bigcap_{i = 1}^k \cI_i$. Next, we define the objective function $f : 2^\cN \rightarrow \mathbb{R}^+$, as follows.
We first define the following function $f'$.
\[
	f'(S) = |\{u_T \in S \mid T \in \cT\}|
	\enspace.
\]
Let $\hat{T} = \{(i, 0)\}_{i = 1}^k$ (note that $\hat{T} \in \cT$). Then,
\[
	f(S)
	=
	\begin{cases}
	f'(S) + \ee & \text{if $u_{\hat{T}} \in S$ and $\{v_i\}_{i = 1}^d \subseteq S$} \enspace, \\
	f'(S) & \text{otherwise} \enspace.
	\end{cases}
\]
Since $f'(S)$ is a linear function, $\DDF{f} = d$.

\begin{claim}
Given the above constructed independence system $(\cN, \cI)$ and objective function $f$, Algorithm~\ref{alg:ExtendibleSystemGreedyDependency} outputs a solution of value $1 + \ee$.
\end{claim}
\begin{proof}
At the first iteration, it is clear that Algorithm~\ref{alg:ExtendibleSystemGreedyDependency} picks exactly the elements of $\{v_i\}_{i = 1}^d + u_{\hat{T}}$, since $\{v_i\}_{i = 1}^d$ is the dependency set of $u_{\hat{T}}$, and the marginal contribution of any other element is at most~1, given any subset of $\cN$.

To complete the proof, we show that Algorithm~\ref{alg:ExtendibleSystemGreedyDependency} cannot increase the value of its solution at the next iterations. Consider an arbitrary element $u \in \cN \setminus (\{v_i\}_{i = 1}^d + u_{\hat{T}})$. If $u = v_x$ for some $0 \leq x < k(d + 1)$, then the addition of $v_x$ does not affect the value of $f$. On the other hand, if $u = u_T$ for some $T \in \cT$, then $T$ must contain a pair $(i, x)$ such that $0 \leq x \leq d$. There are two cases:
\begin{compactitem}
	\item If $x \neq 0$, then $u_T$ cannot coexist in an independent set of $\cM_i$ with $v_x$.
	\item If $x = 0$, then $u_T$ cannot coexist in an independent set of $\cM_i$ with $u_{\hat{T}}$ because both correspond to sets containing the pair $(i, 0)$. \qedhere
\end{compactitem}
\end{proof}

To prove Proposition~\ref{p:tight_extendible_dependency}, we still need to show that $(\cN, \cI)$ contains an independent set of a high value. Consider the set $S^* = \{u_{T^*(j)}\}_{j = 1}^{k(d + 1)}$, where $T^*(j) = \{(i, x) \mid 1 \leq i \leq k \text{ and } x = (i(d + 1) - j) \bmod k(d + 1)\}$.

\begin{claim}
$S^* \subseteq \cN$, hence, $f(S^*) = |S^*| = k(d + 1)$, because $S^* \cap \{v_x\}_{x = 0}^{k(d + 1) - 1} = \varnothing$.
\end{claim}
\begin{proof}
We need to show that for every $1 \leq j \leq k(d + 1)$, $u_{T^*(j)} \in \cN$. 
Let $i = \lceil j / (d+1) \rceil$. Clearly $1 \leq i \leq k$ and $T^*(j)$ contains the pair $(i, x)$ for:
\[
	x
	=
	(i(d+1) - j) \bmod k(d+1)
	=
	(\lceil j / (d+1) \rceil \cdot (d+1) - j) \bmod k(d+1)
	\enspace.
\]
To prove the claim, we need to show that $0 \leq x \leq d$. This follows since $\lceil j / (d+1) \rceil \cdot (d+1) - j \geq (j / (d+1)) \cdot (d+1) - j = 0$ and $\lceil j / (d+1) \rceil \cdot (d+1) - j < [j/(d+1) + 1] \cdot (d+1) - j = d + 1$.
\end{proof}

\begin{claim}
For every two values $1 \leq j_1 < j_2 \leq k(d + 1)$, $T^*(j_1) \cap T^*(j_2) = \varnothing$. Hence $S^* \in \cI$.
\end{claim}
\begin{proof}
Assume towards contradiction that $(i, x) \in T^*(j_1) \cap T^*(j_2)$. Then, modulo $k(d + 1)$, the following equivalence must hold:
\[
	(i(d + 1) - j_1) \equiv (i(d + 1) - j_2)
	\Rightarrow
	j_1 \equiv j_2
	\enspace,
\]
which is a contradiction since $j_1 \neq j_2$ and they are both in the range $[1, k(d + 1)]$.
\end{proof}

\end{document}